\documentclass[conference]{IEEEtran}

\if CLASSOPTIONcompsoc
  \usepackage[nocompress]{cite}
\else
  \usepackage{cite}
\fi

\usepackage[utf8x]{inputenc}
\usepackage{listings}
\usepackage{color}
\usepackage{epsfig}
\usepackage{url}
\usepackage{amsmath}
\usepackage{cite}
\usepackage{graphicx}
\usepackage{color}
\usepackage[bookmarks=false]{hyperref}
\usepackage{amssymb,amsmath,dsfont}
\usepackage{amsthm}
\usepackage[linesnumbered,lined, ruled]{algorithm2e}
\usepackage{epstopdf}
\usepackage{lipsum}
\usepackage{bbm}
\usepackage{tablefootnote}

\usepackage{tikz}
\usetikzlibrary{automata,arrows,positioning,calc}

\usepackage{multirow}

\usepackage{breqn}
\SetAlFnt{\footnotesize}

\usepackage{graphics}

\newtheorem{theorem}{Theorem}[section]
\newtheorem{corollary}{Corollary}[theorem]
\newtheorem{lemma}{Lemma}[section]

\theoremstyle{remark}
\newtheorem*{remark}{Remark}

\usepackage{csquotes}
\usepackage{amsmath} 
\usepackage{amssymb}
\usepackage{subcaption}
\usepackage{graphicx}
\usepackage{accents}
\usepackage{mathtools}
\usepackage[normalem]{ulem}

\SetCommentSty{mycommfont}

\graphicspath{figures/}

\usepackage{mathtools}

\SetKwRepeat{Repeat}{repeat}{until}
\SetKwRepeat{Forever}{repeat}{forever}
\SetKwRepeat{On}{on}{end}
\SetNlSty{bfseries}{\color{black}}{}

\begin{document}

\title{Latency and Backlog Bounds in Time-Sensitive Networking with Credit Based Shapers and Asynchronous Traffic Shaping}


\author{\IEEEauthorblockN{Ehsan Mohammadpour, Eleni Stai, Maaz Mohiuddin, Jean-Yves Le Boudec\\}
\IEEEauthorblockA{\'Ecole Polytechnique F\'ed\'erale de Lausanne, Switzerland\\
$\{$firstname.lastname$\}$@epfl.ch}}

\maketitle

\begin{abstract}
We compute bounds on end-to-end worst-case latency and on nodal backlog size for a per-class deterministic network that implements Credit Based Shaper (CBS) and Asynchronous Traffic Shaping (ATS), as proposed by the Time-Sensitive Networking (TSN) standardization group.
ATS is an implementation of the Interleaved Regulator, which reshapes traffic in the network before admitting it into a CBS buffer, thus avoiding burstiness cascades. 
Due to the interleaved regulator, traffic is reshaped at every switch, which allows for the computation of explicit delay and backlog bounds.
Furthermore, we obtain a novel, tight per-flow bound for the response time of CBS, when the input is regulated, which is smaller than existing network calculus bounds. We also compute a per-flow bound on the response time of the interleaved regulator. Based on all the above results, we compute bounds on the per-class backlogs. Then, we use the newly computed delay bounds along with recent results on interleaved regulators from literature to derive tight end-to-end latency bounds and show that these are less than the sums of per-switch delay bounds.

%
%
\end{abstract}

\IEEEpeerreviewmaketitle

\section{Introduction}
\label{sec:intro}

Time-Sensitive Networking (TSN) is an emerging IEEE standard of the 802.1 Working Group which defines mechanisms for bounded end-to-end latency and zero packet loss \cite{_time-sensitive_task}. It specifies a number of per-class queuing, scheduling and shaping mechanisms. Because the mechanisms are per-class, one key issue in this context is how to deal with the burstiness cascade: individual flows that share a resource dedicated to a class may see their burstiness increase, which may in turn cause increased burstiness to other flows downstream of this resource. Computing latency upper bounds for per-class networks is difficult, unless flows are reshaped at every hop \cite{charny2000delay,bennett2002delay,boyer2008tightening,bouillard2012exact}. This is why a TSN proposal is to reshape flows at every hop, using the concept of interleaved regulator introduced in \cite{specht_urgency-based_2016} and analyzed in \cite{le_boudec_theory_2018} (called ``Asynchronous Traffic Shaping", ATS, within TSN). An interleaved regulator reshapes individual flows 
without per-flow queuing.

In \cite{specht_urgency-based_2016}, an end-to-end delay bound is computed for a network of FIFO constant rate servers with aggregate multiplexing that uses interleaved regulators to avoid the burstiness cascade. However, this does not account for the multi-class nature of a TSN network and for a representative combination of queuing and scheduling mechanisms proposed by TSN, specifically for the scheme called Credit Based Shaper (CBS). The first goal of this paper is to extend these calculations to a more generic TSN network. However, the calculations in \cite{specht_urgency-based_2016} are very complex; extending them seems to be intractable unless some higher level of abstraction is used, as described below. The second goal of this paper is to provide backlog bounds, which can be used to dimension buffers.

To address these goals, we use classic network calculus concepts such as a service-curve characterization of CBS and extend the results in \cite{azua_complete_2014} to include high-priority control-data traffic (CDT). We combine this with the max-plus representation of interleaved regulators proposed in \cite{le_boudec_theory_2018}. Further, we use the result of~\cite{le_boudec_theory_2018} that the upper bound on the delay in the combination of an interleaved regulator following a FIFO system is no greater than the upper bound on the delay of the FIFO system. Overall, in this paper we compute delay upper bounds for the CBS, the interleaved regulator and end-to-end delay bounds along with backlog bounds for the first two. Our main contributions are listed below.
	
i) We obtain a service curve for every AVB class at a CBS system, extending a similar result in \cite{azua_complete_2014} by accounting for the presence of CDT (Theorem \ref{thm:cbfs_delay}). The service curves are used to decouple the interleaved regulator from CBS and are essential to obtain the other results mentioned below.

ii) We obtain a novel, tight bound for the \emph{response} time at a CBS subsystem when the input traffic is reshaped by an interleaved regulator (Theorem \ref{thm:response_CBFS}).
    	
iii) Using this bound and that an interleaved regulator does not increase the delay bound of a FIFO system~\cite{le_boudec_theory_2018}, we obtain a delay bound for the interleaved regulator (Theorem \ref{thm:response_time}).
    	
iv) We use the delay bound of the interleaved regulator to derive a service curve for the interleaved regulator and hence a backlog bound at the interleaved regulator.

v) We are the first to compute a tight end-to-end latency bound for a TSN network of this kind. We show that the end-to-end latency bound obtained is less than the sum of delay bounds computed at every switch along the path of a flow. Ignoring this, as is often done, leads to a gross overestimation of the worst-case end-to-end latency.
Section~\ref{sec:model} describes the system model. Section~\ref{sec:delay} provides: a service curve for the CBS subsystem; a novel tight bound on the response time in the CBS subsystem; a delay bound for the interleaved regulator; and a tight end-to-end delay bound. Section~\ref{sec:buffer} uses these results to derive backlog bounds. Section \ref{sec:example} provides case studies, shows the tightness of the bounds and the sub-additivity of the end-to-end delay bound. Section \ref{sec:conclusion} concludes the paper. 
\section{System Model}
\label{sec:model}

We consider a network with a set $\mathcal{S}$ of nodes (switches and hosts) along with a set of flows, $F$, between hosts. 
Hosts are sources or destinations of flows. There are four types of flows, namely, control-data traffic (CDT), class A, class B, and best effort (BE) \cite{thangamuthu_analysis_2015} in decreasing order of priority. Flows of classes A and B are together referred to as Audio-Video Bridging (AVB) flows, as mentioned in \cite{avb_ieee_2011,azua_complete_2014}. We focus on delay and backlog bounds for AVB traffic. We assume a subset of TSN functions as described next.
%
%

\subsection{Architecture of a TSN node}
\label{subsec:architecture}
\begin{figure}
		\centering
		\includegraphics[width=0.65 \linewidth]{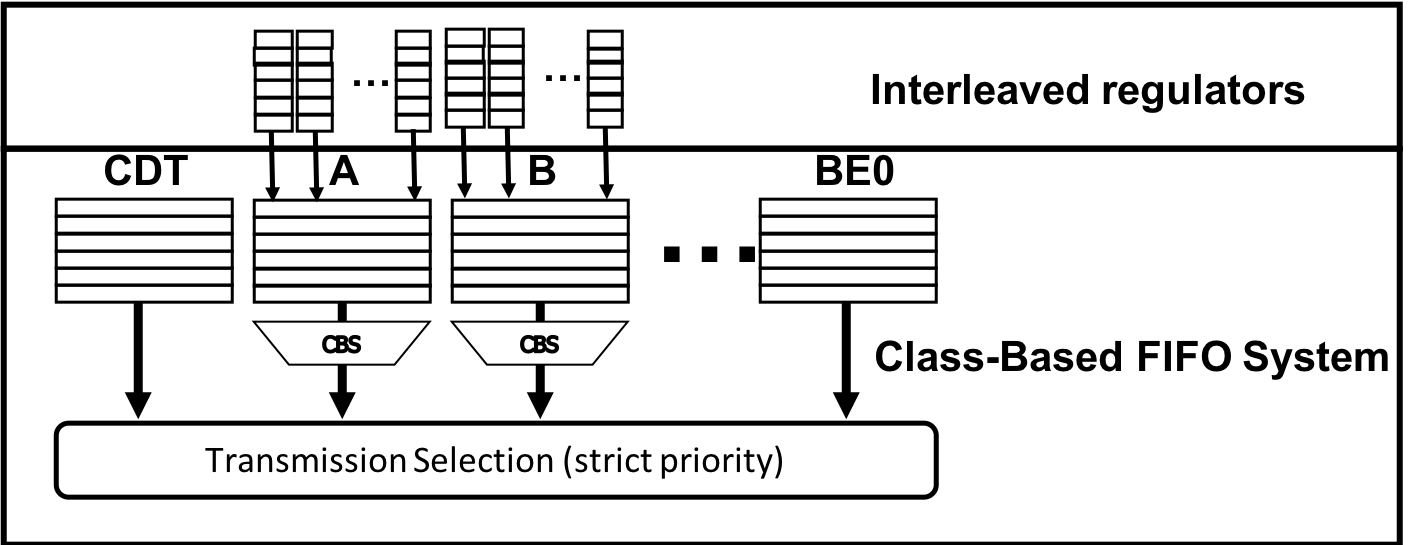}
		\caption{Architecture of one TSN node output port.}
		\label{fig:TSN_switch}
	\end{figure}

We assume that contention occurs only at the output port of a TSN node. Each node output port performs per-class scheduling with eight classes: one for CDT, one for class A traffic, one for class B traffic, and five for BE traffic denoted as BE$_0$-BE$_4$ (TSN standard \cite{_time-sensitive_task}). In addition each node output port also performs per-flow regulation for AVB flows using an interleaved regulator. Thus, at each output port of a node, there is one interleaved regulator per-input port and per-class  \cite{le_boudec_theory_2018,specht_urgency-based_2016}. The detailed picture of scheduling and regulation at a node output port is given by Fig. \ref{fig:TSN_switch}. 
The packets received at a node input port for a given class are enqueued in the respective interleaved regulator at the output port. Then, the packets from all the flows, including CDT and BE flows, are enqueued in a class based FIFO system (CBFS). 

The CBFS includes two CBS subsystems \cite{_ieee_qbv}, one for each class A and B. As defined in \cite{_time-sensitive_task,_ieee_qbv}, the CBS serves a packet from a class according to the available credit for that class. The CDT and BE$_0$-BE$_4$ flows in the CBFS are served by separate FIFO subsystems. Then, packets from all flows are served by a transmission selection subsystem that serves packets from each class based on its priority. All subsystems are non-preemptive.

Guarantees for AVB traffic can be provided only if CDT traffic is bounded; we assume that the CDT traffic from node $i$ to node $j$ has an affine arrival curve $r_{ij} t +b_{ij}$. How to derive such arrival curves involves other TSN mechanisms and is outside the scope of this paper.

\begin{figure}
	\centering
	\includegraphics[width=0.8\linewidth]{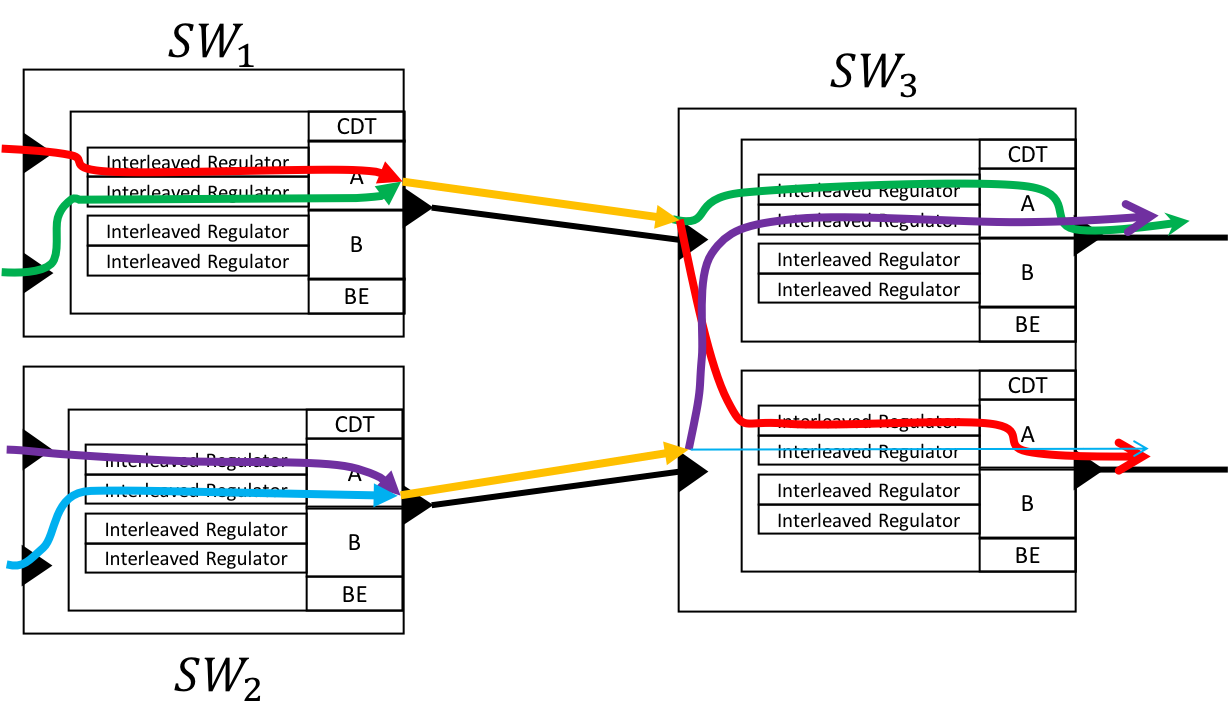}
	\caption{Illustration of the queuing policy by TSN switches for four flows of class A.}
	\label{fig:illustration}
\end{figure}

Fig. \ref{fig:illustration} shows a part of a TSN network with three switches serving four flows of class A. In switches $SW_1$ and $SW_2$ two flows are coming from two different input ports, thus, they use different interleaved regulators. The flows entering switch $SW_3$ from switch $SW_1$ are going to different output ports, and use different interleaved regulators.


\subsection{Flow Regulation}
\label{subsec:interleaved_regualtors}

Following \cite{specht_urgency-based_2016}, we assume that flows are regulated at their source, according to either leaky bucket (LB) or length rate quotient (LRQ).
The LB-type regulation forces flow $f$ to conform to the arrival curve $r_f t+b_f$. The LRQ-type regulation with rate $r_f$ ensures that the time separation between two consecutive packets of sizes $l_n$ and $l_{n+1}$ is at least $l_{n}/r_f$. Note that if flow $f$ is LRQ-regulated, it satisfies the arrival curve constraint $r_f t+L_f$ where $L_f$ is its maximum packet size (but the converse may not hold). For an LRQ regulated flow we set $b_f=L_f$. We also call $M_f$ the minimum packet size of flow $f$. We assume that, at the source hosts, the traffic satisfies its regulation constraint, i.e. we can ignore the delay due to interleaved regulator at hosts.

At every switch inside the network, flows are reshaped inside the interleaved regulator, as described in \cite{le_boudec_theory_2018}. The regulation type and parameters for a flow are the same at its source and at all switches along its path.

\vspace{-0.1in}
\subsection{Other Notation and Definitions}
\label{subsec:notation}
\begin{figure}
	\centering
	\includegraphics[width=1 \linewidth]{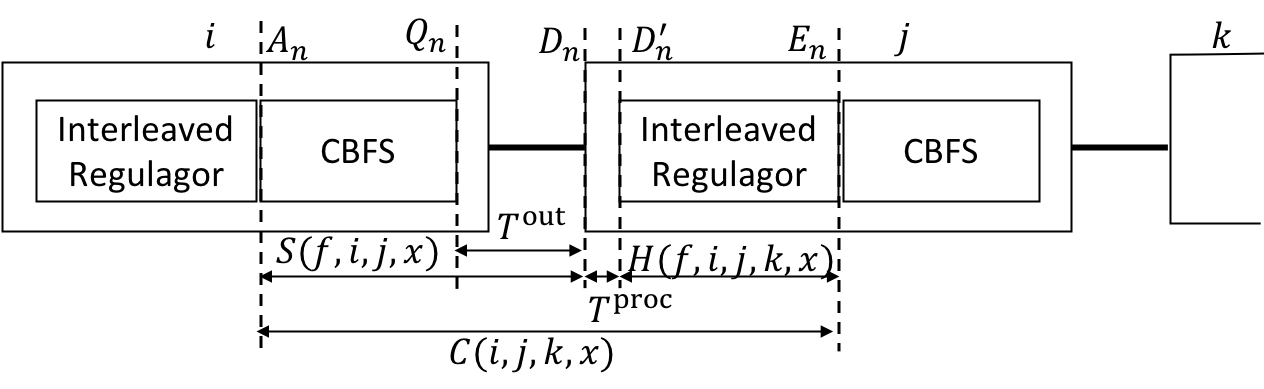}
	\caption{Timing Model in TSN}
	\label{fig:delay_details}
\end{figure}

The indices for nodes, e.g., $i,j,k$ lie in $\left[1, |\mathcal{S}| \right]$.
A directed link from node $i$ to $j$ is denoted by $(i, j)$ with a capacity of $c_{ij}$.
Also, $n \in \mathbb{N}\setminus 0$ is used as an index for packets, $f \in F$ is used as an index for flows, and $x \in \{A, B, E\}$ is used as an index for AVB classes $A$ and $B$, and BE flows, respectively. The set of packets belonging to flow $f$ is $N_f$. The set of flows of class $x$ going from node $i$ to node $j$ is denoted by $F_{ij}^{x}$ and those that continue to node $k$, by $F_{ijk}^{x}$. The flows in $F_{ijk}^{x}$ use the same CBFS in node $i$ and interleaved regulator in node $j$. As mentioned in Section \ref{subsec:architecture}, for each output port, there is a per-class per-input-port interleaved regulator. Thereby, the interleaved regulator in node $j$ connected to link $(j,k)$ indicates an output port of node $j$ connected to node $k$. The maximum packet size of class $x$ flows from node $i$ to $j$ is given by $L_{ij}^{x}$.

Fig. \ref{fig:delay_details} shows the various delays of a packet $n$ of a flow in $F_{ijk}^{x}$. We see five important time instants: (1) $A_n$ is the arrival time of packet $n$ in CBFS, (2) $Q_n$ is the time that packet $n$ starts transmission from CBFS, (3) $D_n$ is the time that packet $n$ is received at a node, (4) $D'_n$ is the time that packet $n$ is enqueued in the interleaved regulator, and (5) $E_n$ is the time that packet $n$ leaves the interleaved regulator.

$(D'_n-D_n)$ is the processing time at node $j$, which is defined as the delay from the reception of the last bit of a packet, coming from node $i$, to the time the packet is enqueued at the interleaved regulator. We assume that $D'_n-D_n \in \left[ T^{\text{proc,min}}_{ij}, T^{\text{proc,max}}_{ij} \right]$. $(D_n-Q_n)$ is the output delay for packet $n$ traversing form node $i$ to node $j$, which is defined as the time required from the selection of a packet for transmission from a CBFS queue of node $i$ to the reception of the last bit of the packet by the node $j$. Also, $D_n-Q_n=l_n/c_{ij}+T^{\text{var},n}_{ij}$, where $l_n$ is the length of packet $n$  
and  $T^{\text{var},n}_{ij}$ 
is in $\big[T^{\text{var, min}}_{ij}, T^{\text{var, max}}_{ij}\big]$.

We compute the following bounds for packets of AVB flows belonging to class $x$ going from node $i$ to $j$ (see Fig. \ref{fig:delay_details}):


	$\bullet$ $S(f,i,j,x)$: upper bound on the response time for flow $f$ in CBFS, i.e. on $(D_n-A_n)$;
%

	$\bullet$ $H(f,i,j,k,x)$: upper bound on the response time for flow $f$ in the interleaved regulator at node $j$ 's output port for link $(j,k)$, i.e. a bound on $(E_n-D'_n )$; 

$\bullet$ $C(i,j,k,x)$: upper bound for all flows on the response time in the combination of the CBFS at node $i$ and the interleaved regulator at node $j$ for link $(j,k)$, i.e. a bound on $(E_n-A_n)$ for all flows.

\vspace{-0.01in}
\section{Delay Bounds in TSN}\label{sec:delay}
The aim of this section is the computation of bounds on the delays an AVB flow experiences due to CBFS, $S(f,i,j,x)$, and interleaved regulator at a node, $H(f,i,j,k,x)$. To do so, in Section \ref{subsec:service_cbfs}, we first derive a service curve of CBFS for an AVB flow, in presence of CDT with an LB arrival curve. Then, in Section \ref{sec:FIFO_delay}, we use this service curve to compute a bound on the response time for an AVB flow in the CBFS of a node, i.e., $S(f,i,j,x)$. Using $S(f,i,j,x)$, we compute $C(i,j,k,x)$. Consequently, we can compute the bound on the delay of the interleaved regulator, $H(f,i,j,k,x)$ in Section \ref{sec:intreg}, and therefore we have all the elements to compute a bound on the delay of a single TSN node. We also compute a tight end-to-end delay bound for an AVB flow in Section \ref{subsec:e2e}.

 \vspace{-0.02in}
\subsection{Service Curve Offered by CBFS to AVB flows}
\label{subsec:service_cbfs}
The following theorem provides service curves offered by a CBFS at a TSN node, for AVB flows in presence of CDT flows with LB arrival curve. In \cite{azua_complete_2014}, the authors compute service curves for AVB flows according to the IEEE AVB standard \cite{avb_ieee_2011}, i.e., in absence of CDT. Note that service curves for AVB flows in TSN are proposed in \cite{zhao_timing_analysis_2018}; however in their proof credit reset is not considered. We obtain different service curves than \cite{zhao_timing_analysis_2018} and we use them to obtain tight delay bounds.
\vspace{-0.02in}
	\begin{theorem}
	\label{thm:cbfs_delay}
	Assume a node $i$ and a link $(i,j)$, where the CDT has an LB arrival curve with parameters $(r_{ij}, b_{ij})$ and the line rate is $c_{ij}$. Then, the CBFS offers to class $A$ flows a rate-latency service curve with parameters,
	
	\begin{align}\label{eq:service_curve_A}
	T_{ij}^A &= \frac{1}{c_{ij}-r_{ij}}\Big(\bar{L}^A_{ij}+b_{ij}+\frac{r_{ij}\bar{L}_{ij}}{c_{ij}}\Big), \\
	R_{ij}^A &= 	\frac{I^A_i (c_{ij}-r_{ij}) }{I^A_{ij}-S^A_{ij}},
	\end{align}\noindent where $I^A_{ij}$ and $S^A_{ij}$ are the \emph{idle slope} and \emph{send slope}, correspondingly, of the CBS for class $A$ and link $(i,j)$, $\bar{L}^A_{ij} = \max(L^B_{ij}, L^E_{ij})$, and $\bar{L}_{ij} = \max(L^A_{ij},L^B_{ij}, L^E_{ij})$. Similarly for class $B$ flows, CBFS offers a rate-latency service curve with parameters,
	\begin{align}\label{eq:service_curve_B}
	T_{ij}^B &= \frac{1}{c_{ij}-r_{ij}}\Biggl(L^E_{ij}+L^A_{ij}-\frac{\bar{L}^A_{ij} I^A_{ij}}{S^A_{ij}} +b_{ij}+\frac{r_{ij}\bar{L}_{ij}}{c_{ij}}\Biggr), \\
	R_{ij}^B &= 	\frac{I^B_{ij} (c_{ij}-r_{ij}) }{I^B_{ij}-S^B_{ij}},
	\end{align}\noindent where $I^B_{ij}$ and $S^B_{ij}$ are the \emph{idle slope} and \emph{send slope}, correspondingly, of the CBS for class $B$ and link $(i, j)$.
\end{theorem}
The proof is given in the Appendix \ref{app:1}.

\subsection{Upper Bound on the Response Time in CBFS} \label{sec:FIFO_delay}
The rate-latency service curve offered by CBFS at node $i$ for link $(i,j)$ to class $x\in \{A,B\}$ has parameters $R_{ij}^x$, $T_{ij}^x$, as calculated in Theorem \ref{thm:cbfs_delay}. Also, let $b_{ij}^{\mathrm{tot},x}= \sum_{f \in F_{ij}^{x}}b_f$. Then, the following theorem gives an upper bound on the response time for a flow $f$ at a CBFS of node $i$.

\begin{theorem}\label{thm:response_CBFS}
A tight upper bound on the response time in the CBFS of node $i$ (following the interleaved regulator) for flow $f$ of class $x\in \{A,B\}$, going from node $i$ to $j$, is:
\begin{align} \label{eq:CBFS_response_time}
S(f,i,j,x) = T_{ij}^x + \frac{b_{ij}^{\mathrm{tot},x}-\psi_{f}}{R_{ij}^x} + \frac{\psi_{f}}{c_{ij}} + T_{ij}^{\mathrm{var,max}},
\end{align}\noindent where the parameter $\psi_{f}$ depends on the flow $f$ and the type of regulator, namely, for LRQ: $\psi_{f} = L_f$ and for LB: $\psi_{f} = M_f$. 
\end{theorem}
The proof is given in the Appendix \ref{app:2}.

\begin{remark}
Importantly, we should note that the bound on the response time in the CBFS given by Eq. \eqref{eq:CBFS_response_time} improves the corresponding bound obtained by using the classical network calculus approach (see \cite{le_boudec_network_2001}, Theorem 1.4.2 and Section 1.4.3). Specifically, the latter bound is not a per-flow bound and is equal to $T_{ij}^x + \frac{b_{ij}^{\mathrm{tot},x}}{R_{ij}^x} + T_{ij}^{\mathrm{var,max}}$, which is always larger than $S(f,i,j,x)$ (Eq. \eqref{eq:CBFS_response_time}) since $\frac{-\psi_{f}}{R_{ij}^x} + \frac{\psi_{f}}{c_{ij}}<0$. We reached this improved bound by combining the min-plus representation of service curve and max-plus representation of regulation \cite{le_boudec_theory_2018}.
\end{remark}

It is known from \cite{le_boudec_theory_2018} that for all flows belonging to class $x$ sharing the same CBFS queue at node $i$ and interleaved regulator at node $j$ (e.g., for link $(j,k)$),
\begin{align} \label{eq:pi_reg_del}
C(i,j, k, x) = \sup_{f'\in F_{ijk}^x}S(f',i,j,x)+ T_{ij}^{\mathrm{proc}, \max}.
\end{align}
Therefore, the following Corollary is a direct result.
\begin{corollary}\label{col:combined_response_time}
Assume flows of class $x\in \{A,B\}$, going from node $i$ to $j$, and enqueued in the interleaved regulator at node $j$ for link $(j,k)$. An upper bound, for each flow, of the combination of the response time in  CBFS of node $i$ (following the interleaved regulator of $i$) and the interleaved regulator at node $j$ for link $(j,k)$ is given by:
	\begin{align}\label{eq:col_combined_del_bound}
		C(i,j,k,x) =& T^x_{ij} + \frac{b_{ij}^{\mathrm{tot},x}}{R^x_{ij}}+T^{\mathrm{var,max}}_{ij} \nonumber \\
&+\sup_{f'\in F^x_{ijk}}\Bigg(\frac{\psi_{f'}}{c_{ij}}-\frac{\psi_{f'}}{R^x_{ij}}\Bigg)+ T_{ij}^{\mathrm{proc}, \max},
	\end{align}\noindent where for LRQ: $\psi_{f} = L_f$ and for LB: $\psi_{f} = M_f$.
\end{corollary}

\subsection{Bound on the Response Time in the Interleaved Regulator}
\label{sec:intreg}
The following theorem proves an upper bound on the response time in the interleaved regulator, $H(f,i,j,x)$.

\begin{theorem}\label{thm:response_time}
	An upper bound on the response time for flow $f$ of class $x \in \{A,B\}$ in the interleaved regulator at node $j$ for link $(j,k)$ that follows the CBFS of node $i$ is:
	\begin{align}\label{eq:interleaved_del_bound}
	H(f,i,j,k,x) =  C(i,j,k,x) - \frac{M_f}{c_{ij}} - T_{ij}^{\mathrm{var,min}}-T_{ij}^{\text{proc, min}} .
	\end{align}
\end{theorem}
The proof is given in the Appendix \ref{app:3}.

%

\begin{remark}
It is shown numerically in Section \ref{sec:example}, that $H$ is tight for the flow $f$ that acheives the maximum response time at the CBFS, i.e., for which $S(f,i,j,x)=C(i,j,k,x)$.
\end{remark}
\vspace{-0.05in}
\subsection{Upper Bound on the End-to-End Delay}
\label{subsec:e2e}
Assume an AVB flow $f$ routed through the nodes $(i_1,...,i_k)$, where the source is $i_1$ and destination is $i_k$. It is assumed that the arrival curves of the generated flows in source conform to the flows' regulation policies, and thus the flows do not experience delay at the interleaved regulators of the source nodes. An upper bound on the end-to-end delay for flow $f$ of class $x$, namely, $D_f^x$, is,
\vspace{-0.08in}
\begin{align}\label{eq:end_to_end_latency_bound}
D_f^x &= \sum_{j=1}^{k-2}C(i_j,i_{j+1},i_{j+2},x)+ S(f,i_{k-1},i_k,x).
\end{align} \noindent 

$D_f^x$ can be easily computed by using Eqs. (\ref{eq:CBFS_response_time}), \eqref{eq:col_combined_del_bound}. In Section \ref{subsec:tightness_e2e}, we show numerically that this bound is tight.


\vspace{-0.06in}
\section{Backlog Bounds}
\label{sec:buffer}
In this section, we determine an upper bound on the backlog for each AVB class of interleaved regulator and CBFS.

\vspace{-0.06in}
\subsection{Backlog Bound on Interleaved Regulator}
In network calculus, computing upper bounds on the backlog requires information on arrival and service curves \cite{le_boudec_network_2001}.
\vspace{-0.1in}
\subsubsection{Service Curve Offered by Interleaved Regulator}
It is known that a service curve offered by a FIFO system which guarantees a maximum delay $D$, is equal to the ``impulse" function $\delta_D(t)$\footnote{defined as $\delta_D(t)=0$, $0\leq t\leq D$, else  $\delta_D(t)=+\infty$.} \cite{le_boudec_theory_2018}. The interleaved regulator is a FIFO system, for which a delay upper bound is computed in Section \ref{sec:intreg}. Therefore, a service curve offered by the interleaved regulator for class $x$, at node $j$ for link $(j,k)$, that follows a CBFS of node $i$ is $\delta_{D(i,j,k,x)}(t)$, where $D(i,j,k,x)=\sup_{f' \in F_{ijk}^x}{H(f',i,j,l,x)}$ is computed using Theorem \ref{thm:response_time}.
\vspace{-0.01in}
\subsubsection{Arrival Curve of Interleaved Regulator Input} \label{sec:acinterl}

The output flows of the upstream CBFS (node $i$) may not share the same interleaved regulator. Let us consider the interleaved regulator of node $j$ for link $(j,k)$ that follows the CBFS of node $i$. Suppose that $r_{s}$ and $b_{s}$ are the sum of rates and bursts of the flows $f' \in F_{ijk}^x$ for $x \in \{A,B\}$. In addition, $r_{w}$ and $b_{w}$ are the sum of rates and bursts of the flows that do not use the same interleaved regulator in downstream node with the previous flows. The CBFS offers a rate-latency service curve with parameters $(R_{ij}^x,T_{ij}^x)$ to the class $x \in \{A,B\}$ (Theorem \ref{thm:cbfs_delay}). Then, according to \cite{le_boudec_network_2001}, the output arrival curve of the former flows is an LB one, $r_s t +b_{out}$ with $b_{out}=b_{s} + r_{s}(T_{ij}^x+\frac{b_{w}}{R_{ij}^x})$.

On the other hand, the upstream line has constant rate, $c_{ij}$. Therefore, it also enforces an arrival curve to the input of the interleaved regulator equal to $c_{ij}t + \sup_{f'\in F_{ijk}^x}{L_{f'}}$.

As the CBFS follows the interleaved regulator, the input arrival curve of the interleaved regulator is,
\begin{equation}
\alpha(t) = \min\Big(c_{ij}t + \underset{f'\in F_{ijk}^x}{\text{sup}}\{L_{f'}\},r_{s}t+b_{s} + r_{s}(T_{ij}^x+\frac{b_{w}}{R_{ij}^x})\Big).
\end{equation}
\vspace{-0.1in}
\subsubsection{Backlog Bound on Interleaved Regulator}
The backlog bound is calculated as 
$\sup_{s\geq 0}{\alpha(s) - \beta(s)}$
\cite{le_boudec_network_2001}, where $\alpha(t)$ is the arrival curve and $\beta(t)$, the service curve. By replacing the arrival and service curves obtained in the two previous subsections, we obtain the backlog bound of the interleaved regulator for class $x \in \{A,B\}$ at node $j$ for link $(j,k)$ that follows the CBFS of node $i$, denoted as $B^{\text{IR}, x}_{ijk}$ and given:
\begin{equation}\label{eq:IR_buffer_bound}
\begin{split}
B^{\mathrm{IR}, x}_{ijk}= \min&\Big(c_{ij}D(i,j,k, x) + \sup_{f'\in F_{ijk}^x}\{L_{f'}\},\\
&r_{s}D(i,j,k,x)+b_{s} + r_{s}(T_{ij}^x+\frac{b_{w}}{R_{ij}^x})\Big),
\end{split}
\end{equation}\noindent where 
$T_{ij}^x$, $R_{ij}^x$ are computed in Theorem \ref{thm:cbfs_delay}.

\vspace{-0.08in}
\subsection{Backlog Bound on Class-Based FIFO System}
Consider all flows $f \in F_{ij}^{x}$. The input of the CBFS for class $x$ has an arrival curve equal to the sum of all $\alpha_f$. 
Using Theorem \ref{thm:cbfs_delay} and following a process similar to the one followed for the interleaved regulator, the backlog bound of the CBFS at node $i$ for link $(i,j)$ and class $x$, denoted as $B_{i,j}^{\text{CBFS},x}$ is

\begin{equation}\label{eq:CBFS_buffer_bound}
B_{i,j}^{\mathrm{CBFS},x} = \sum_{f' \in F^x_{ij}}b_{f'}+\sum_{f' \in F^x_{ij}}r_{f'} T^x_{ij}.
\end{equation}

\section{Case Study}
\label{sec:example}
\begin{figure}
	\centering
	\includegraphics[width=0.4 \linewidth]{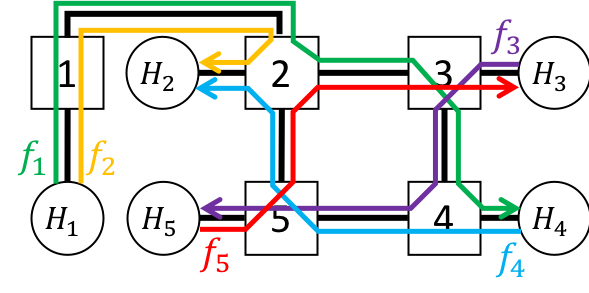}
	\caption{Practical TSN network used for the case study.}
	\label{fig:e2e_delay_net}
\end{figure}

In this section, we apply the results obtained in the previous sections to practical TSN networks (Fig. \ref{fig:e2e_delay_net}, \ref{fig:e2e_tight_delay_net}). We highlight the tightness of the delay bounds obtained and the sub-additivity property of the end-to-end delay bound. 

\subsection{TSN Network Setup and Flows}
\label{subsec:network_setup}
We use the network shown in Fig. \ref{fig:e2e_delay_net}. It consists of five switches labeled $1$-$5$, and five hosts (as sources and destinations of flows), namely $H_1$-$H_5$, with five class A flows $f_1$-$f_5$. Flow $f_1$ is LRQ regulated with rate $r_{f_1}=20 ~Mbps$ and has maximum packet length $L_{f_1}=1 ~Kb$. Flows $f_2$-$f_5$ are LRQ regulated with rate $20 ~Mbps$ and maximum packet length $2 ~Kb$. It is assumed that on each output port there is a CDT flow with an LB arrival curve $(20 ~Mbps, 4 ~Kb)$, and a BE flow with maximum packet length of $2 ~Kb$.
As is shown in the Fig. \ref{fig:e2e_delay_net}, due to circular dependency among the flows, use of interleaved regulators leads to bounded end-to-end latency for the flows. Otherwise, burstiness cascade of the flows would cause unbounded end-to-end latency.

For ease of presentation, we assume that the CBFS has only three classes: CDT, class A, and one BE. Moreover, $T_{ij}^{\text{var,max}}$ and $T_{ij}^{\text{proc,\text{max}}}$ are zero in all switches $i$, links $(i,j)$ and all packets of a same flow have the same size. The line rate is equal to $100 ~Mbps$. The parameters of CBS are $I^A_{ij}= 50 ~Mbps$ and $S^A_{ij}= -50 ~Mbps$, $\forall ~ i,j, H_1-H_5$. We are interested in studying the worst case response time of flow $f_1$ in CBFS of host $H_1$ and its corresponding interleaved regulator in switch $1$. Also, we compute the theoretical end-to-end delay bound of this flow and show its sub-additivity property.

\begin{figure*}
	  	\centering
	  	\begin{subfigure}{.25\textwidth}
	  		\centering
	  		\includegraphics[width=1\linewidth]{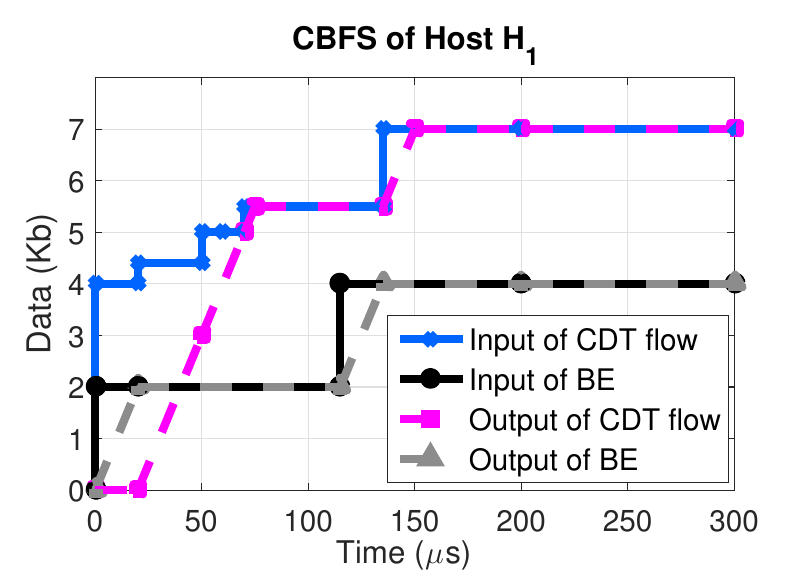}
	  		\caption{CDT and BE flows}
	  		\label{fig:arrival_departure_1}
	  	\end{subfigure}%
	  	\begin{subfigure}{.25\textwidth}
	  		\centering
	  		\includegraphics[width=1\linewidth]{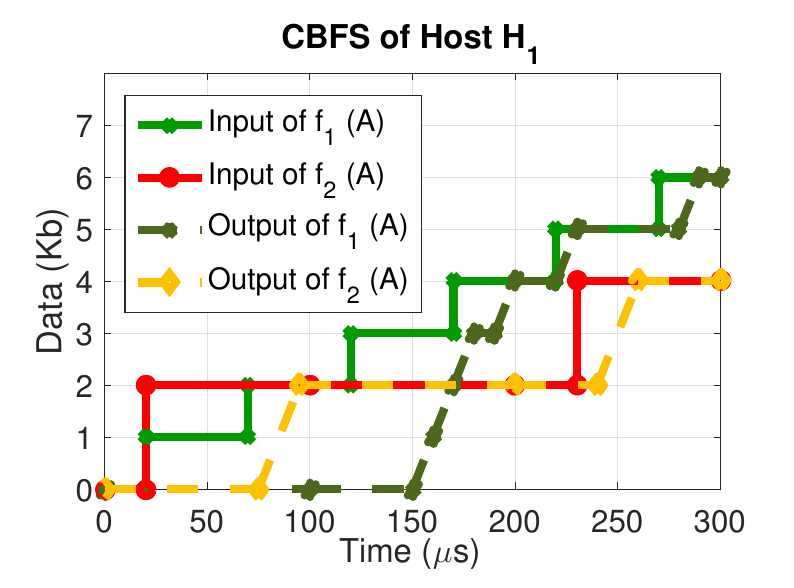}
	  		\caption{$f_1$ and $f_2$ flows (class A)}
	  		\label{fig:arrival_departure_2}
	  	\end{subfigure}
  	\begin{subfigure}{.25\textwidth}
  		\centering
  		\includegraphics[width=1\linewidth]{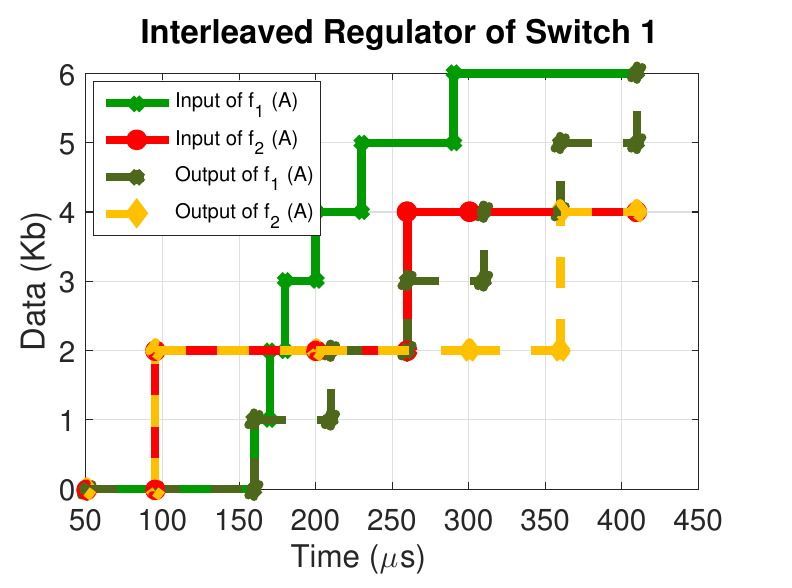}
  		\caption{$f_1$ and $f_2$ flows (class A)}
  		\label{fig:arrival_departure_3}
  	\end{subfigure}
	  	\caption{Cumulative data input and output curves for the CBFS of $H_1$ and interleaved regulator of switch $1$, with respect to Fig. (\ref{fig:e2e_delay_net}). Figures (a) and (b) show the data arrival and departures from the CBFS of $H_1$, and the figure (c) shows the arrival and departure of data from the interleaved regulator for flows $f_1$ and $f_2$ in switch $1$.}
	  	\label{fig:arrival_departure}
	  \end{figure*}

\subsection{Computation of Theoretical Bounds}
\label{subsec:theo_comp}
We compute the obtained upper bounds for the response time in CBFS and interleaved regulator for flow $f_1$, and the backlog bounds for the host $H_1$ and switch $1$. According to Theorem \ref{thm:response_CBFS}, the bound on the CBFS response time for flow $f_1$ in the host $H_1$ is $S(f_1,H_1,1,A) = 140~\mu s$. Also, from Theorem \ref{thm:response_time}, the bound on the response time in interleaved regulator for flow $f_1$, enqueued in the output port for link $(1,2)$ on switch $1$ is $H(f_1,H_1,1,2,A) =130~\mu s$. Also, the backlog bound for the same interleaved regulator (Eq. \eqref{eq:IR_buffer_bound}) is $11.4 ~~Kb$. The backlog bound for CBFS of class A in host $H_1$ is $6.2 ~~Kb$ (Eq. \eqref{eq:CBFS_buffer_bound}). To compute the end-to-end delay, we use Eq. (\ref{eq:end_to_end_latency_bound}). 
Using Eq. \eqref{eq:col_combined_del_bound}, we find that $C(H_1,1,2,A) = C(1,2,3,A) = C(2,3,4,A) = C(3,4,H_4,A) = 140~\mu s$, and $S(f_1,4,H_4,A) = 140~\mu s$. Thus, for flow $f_1$ of class $A$ we have the upper bound on delay $D_{f_1}^A = 700~\mu s$.

\subsection{Numerical Example of Tightness}
\label{subsec:numerical}
Next, we show how these bounds are tight by presenting a particular series of packet arrivals as shown in Fig. \ref{fig:arrival_departure}. This figure shows the input and output curves related to $f_1$, $f_2$, CDT and BE flows in host $H_1$ and switch $1$. A step in the input curve indicates the time of reception of the entire packet. According to Fig. \ref{fig:arrival_departure_1}, at time $0~\mu s$, a packet of BE arrives and starts being transmitted. At time $0^+~\mu s$, a burst of CDT traffic arrives and then for time $t \geq 0^+$, CDT traffic continues to arrive with rate $20 ~Mbps$ up to the time $75~\mu s$. The transmission of CDT traffic at time $0^+$ is blocked by the transmission of the BE packet as all switches are non-preemptive. At time $20~\mu s$, CDT traffic has accumulated a backlog and starts its transmission.

From Fig. \ref{fig:arrival_departure_2}, we see that time $20~\mu s$ is the start of the backlog period of class A since a packet of flow $f_2$ and a packet of $f_1$ arrive, with first of the two being the former. The first packet of flow $f_2$ reaches at time $95~\mu s$ the interleaved regulator in switch $1$ for link ($1,2$) that implies a response time of $75~\mu s$ for flow $f_2$ in the CBFS of host $H_1$. The first packet of flow $f_1$ finishes its transmission at time $160~\mu s$ from CBFS in $H_1$, due to its earlier blockage by the CDT and $f_2$ traffic. This implies a response time of $140~\mu s$ for flow $f_1$ in CBFS of $H_1$, i.e., equal to the bound in Section \ref{subsec:theo_comp}.

 From Fig. \ref{fig:arrival_departure_3}, we notice that the worst-case response time in the interleaved regulator for flow $f_1$ at switch $1$ is for the packet that arrives at time $230~\mu s$. This packet is declared eligible by the interleaved regulator at time $360~\mu s$. This implies the response time of $130~\mu s$ in the interleaved regulator that is the upper bound for flow $f_1$ as computed in Section \ref{subsec:theo_comp}. The maximum response time seen in Fig. \ref{fig:arrival_departure_3} for flow $f_2$ is for its packet that arrives at time $260~\mu s$ at the interleaved regulator at switch $1$ and is equal to $100 ~\mu s$.

Note that this packet of  flow $f_2$ could have been declared eligible by the interleaved regulator already at $260~\mu s$ but is blocked by preceding packets of flow $f_1$ that were not yet eligible at that time. Based on Fig. \ref{fig:arrival_departure_2} and \ref{fig:arrival_departure_3}, we observe that packets of flow $f_1$ experience a maximum delay of $140~\mu s$ from the time being enqueued in the CBFS of $H_1$ to the time being declared eligible by the interleaved regulator at switch $1$ (equal to $C(H_1,1,2,A)$, computed in Section \ref{subsec:theo_comp}).

The maximum observed backlog for class A used in the CBFS at the output port of $H_1$ is equal to $4 ~~Kb$ during times $70~\mu s$ to $75\mu s$, which is $65\%$ of the computed bound. Furthermore, the maximum backlog observed in the interleaved regulator at output port of switch $1$ is equal to $5 ~~Kb$ during times $230~\mu s$ to $260~\mu s$, which is $43\%$ of the computed bound.

\vspace{-0.03in}
\subsection{Sub-additivity of End-to-End Delay Bound}
\label{subsec:sub-additivity}
\vspace{-0.03in}
In TSN, the common way of computing the end-to-end delay bound is by adding the delay bounds of each switch in the path of a flow. However, Eq. (\ref{eq:end_to_end_latency_bound}) provides a much better upper bound. To show this, we first compute the delay upper bound for switch $i$ following switch $j$ and followed by switch $k$ in the path of a flow $f$ of class $x$, is given by,

\begin{align}\label{eq:end_2_end_switch_delay}
d_{j,i,k}^{x,f} =& H(f,{j},{i},{k},x) + S(f,{i},{k},x) + T^{\mathrm{proc,max}}_i,
\end{align}
where for $i$ being a source, $H(f,{j},{i},{k},A)$ is equal to zero. Considering $T^{\text{proc,max}}_i$ equal to zero in this case, an end-to-end delay bound for flow $f_1$ over the path $(H_1,1), (1,2), (2,3), (3,4), (4, H_4)$ can be computed as $(0+140)+ 4\times(130+140) =1220~\mu s$.
%

From Section \ref{subsec:numerical}, we know that an upper bound on the end-to-end delay for flow $f_1$ is $700~\mu s$ which is $57\%$ of $1220 \mu s$, obtained using Eq. \ref{eq:end_2_end_switch_delay}. A similar phenomenon occurs with per-flow networks, called ``pay bursts only once".

\subsection{Tightness of End-to-End Delay Bound}
\label{subsec:tightness_e2e}
\begin{figure}
	\centering
	\includegraphics[width=0.5 \linewidth]{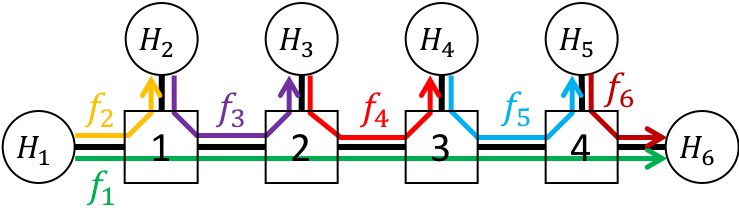}
	\caption{TSN network for tightness of end-to-end delay bound.}
	\label{fig:e2e_tight_delay_net}
\end{figure}
Consider the network shown in Fig. \ref{fig:e2e_tight_delay_net}, having four switches labeled $1$ - $4$, and six hosts, namely $H_1$ - $H_6$, with six class A flows $f_1$ - $f_6$. The assumptions on $f_1-f_5$, CDT and BE traffic are as in Section \ref{subsec:network_setup} and $f_6$ is similar to $f_2-f_5$.

To show the tightness of end-to-end delay bound of Eq. \eqref{eq:end_to_end_latency_bound},  we claim that each pair of $f_1$, $f_3$ at switch $1$, $f_1$, $f_4$ at switch $2$, $f_1$, $f_5$ at switch $3$, and $f_1$, $f_6$ at switch $4$ experience the same input/output curves as the pair of flows $f_1$, $f_2$ in Fig. \ref{fig:arrival_departure} but appropriately shifted in time, so that they take place sequentially.
 Thus, flow $f_1$ has a delay of $140~\mu s$ from the time being enqueued in the CBFS of $H_1$ to the time declared eligible from the interleaved regulator at $1$. The same delay is experienced by flow $f_1$ at the rest pairs of switches in its path. Similar to Section \ref{subsec:numerical}, the response time of $f_1$ at CBFS of switch $4$ is equal to $140~\mu s$. Therefore, the end-to-end delay for $f_1$ is equal to $4\times140+140= 700~\mu s$, which is equal to the bound computed from Eq. \eqref{eq:end_to_end_latency_bound}. 
\section{Conclusion}\label{sec:conclusion}
We have provided a set of formulas for computing bounds on end-to-end delay and backlog for class A and class B traffic in a TSN network that uses CBS and ATS. The bounds are rigorously proven, while we provide a representative case study that highlights the tightness of the delay bounds provided and shows the sub-additivity of the end-to-end delay bound. Future work will address other mechanisms in TSN.

\bibliographystyle{IEEEtran}
\vspace{-0.05in}
\bibliography{ref}

  \appendices

\section{Proof of Theorem \ref{thm:cbfs_delay}}\label{app:1}
\begin{proof}
	In this proof, for the ease of presentation, the indices $i$ and $j$ are neglected from the parameters. For a given class $x \in \{A,B\}$, $V^x(u)$ stands for the amount of CBS credit at time $u$, which is less than maximum credit, namely $V^{x,\text{max}}$. Also, let $N^x(u)$, $O^{x}(u)$ be the cumulative input, output traffic, respectively, for class $x$ at time $u$ and $O^{H}(u)$ the cumulative output CDT traffic. Given a time $t$, we define the time $s$, as follows,
	\begin{equation}
	s = \sup\Big\{u\leq t ~| ~V^x(u)=0 \text{ and } O^{x}(u)=N^x(u)\Big\}.
	\end{equation}
	In the rest of the proof, we use $\Delta t = t-s$ as the time interval over which the system is analyzed. For a given class $x$, we also define the time intervals $\Delta t^{x-}$, $\Delta t^{x+}$, and $\Delta t^{x0}$, as the aggregated time periods within $\Delta t$, over which the credit is decreasing (due to packet transmission of this class), the credit is increasing (due to transmission of BE traffic or AVB traffic except class $x$ in backlog period of this class or credit recovery), and the credit is frozen (due to transmission of CDT flows), respectively. Then,
	\begin{equation}
	V^x(t) - V^x(s) = V^x(t) = \Delta t^{x+} I^x + \Delta t^{x-} S^x.
	\end{equation}
	 As for a class $x\in \{A,B\}$, $\Delta t = \Delta t^{x+} + \Delta t^{x-} + \Delta t^{x0}$, then,
	\begin{equation}\label{eq:proof_service_delta-}
	\Delta t^{x-} = \frac{I^x \Delta t -V^x(t) - I^x \Delta t^{x0}}{I^x - S^x}.
	\end{equation}
	The number of bits served during $\Delta t$ for class $x \in \{A,B\}$ is $c\Delta t^{x-}$. By utilizing Eq. \eqref{eq:proof_service_delta-},
	\begin{equation}\label{eq:proof_service_cdt_out}
	O^{x}(t) - O^{x}(s) = \frac{c}{I^x-S^x}(I^x \Delta t -V^x(t) - I^x \Delta t^{x0}).
	\end{equation}
	$\Delta t^{x0}$ is the aggregated transmission time periods of CDT flows during $\Delta t$; therefore,
	\begin{equation}\label{eq:proof_service_delta0-1}
	\Delta t^{x0} = \frac{O^{H}(t) - O^{H}(s) }{c}.
	\end{equation}
The service curve offered to an aggregate of CDT flows in CBFS at time $t > s$ is $\beta^H(t)=c[t-\bar{L}/c]^+$, where  $\bar{L}=\max(L^A_{ij},L^B_{ij}, L^E_{ij})$. As the CDT queue has an LB arrival curve (denoted as $\alpha^H=rt +b$), the following equation holds for the output arrival curve of the CDT queue:
	\begin{equation}\label{eq:proof_service_cdt_out_arrival}
	O^{H}(t) - O^{H}(s) \leq (\alpha^H \oslash \beta^H)(\Delta t) = b+r(\Delta t+\frac{\bar{L}}{c}).
	\end{equation}
	Therefore, from Eqs. \eqref{eq:proof_service_delta0-1} and \eqref{eq:proof_service_cdt_out_arrival},
	\begin{equation}\label{eq:proof_service_delta0-2}
	\Delta t^{x0} \leq \frac{b+r(\Delta t+\frac{\bar{L}}{c})}{c}.
	\end{equation}
	From Eqs. \eqref{eq:proof_service_cdt_out} and \eqref{eq:proof_service_delta0-2}, and knowing that $O^{x}(s)=N^x (s)$ (from the assumption for $s$):
	\begin{equation} \label{eq:outputx}
	O^{x}(t) - N^x(s) \geq \frac{I^x(c-r)}{I^x-S^x}\Bigg(\Delta t - \frac{V^x(t)c}{I^x(c-r)}-\frac{b+\frac{r\bar{L} }{c}}{c-r}\Bigg).
	\end{equation}\noindent As $V^x(t) \leq V^{x,\text{max}}$, and 
	using Eq. \eqref{eq:outputx}, we define the following rate-latency service curve for class $x \in \{A,B\}$:
	\begin{equation}
	\beta^x(t)= \frac{I^x(c-r)}{I^x-S^x}\Bigg[ t - \frac{V^{x,\mathrm{max}}c}{I^x(c-r)}-\frac{b+\frac{r\bar{L}}{c}}{c-r}\Bigg]^+.
	\end{equation}
	Finally, we replace in the last equation the credit upper bound given in \cite{azua_complete_2014} for each class $x  \in\{A,B\}$. For completeness:
	\begin{align}
	V^{A,\mathrm{max}} &= \bar{L}^A \frac{I^A}{c},\\
	V^{B,\mathrm{max}} &= \frac{I^B}{c}\Bigg(L^E + L^A - \bar{L}^A \frac{I^A}{S^A}\Bigg),
	\end{align} \noindent and we obtain the equations in the statement of Theorem \ref{thm:cbfs_delay}.
\end{proof}

\section{Proof of Theorem \ref{thm:response_CBFS}}\label{app:2}
\begin{proof}
First we compute the bound and then we show that the bound is tight. For the computation of $S(f,i,j,x), \forall f,i,j,x$, we use the waiting time of the $n^{th}$ packet assuming it belongs to the flow $f$ at the CBFS, $W(n,f,i,j,x)$, computed in Lemma~\ref{lemma:waiting_time}:

\begin{align} \label{eq:CBFS_response_2a}
D_n-A_n&=(Q_n-A_n) +(D_n-Q_n)\\
&\leq W(n,f,i,j,x) + D_n-Q_n\\
&\leq W(n,f,i,j,x) + \frac{l_n}{c_{ij}}+T^{\mathrm{var},\max}_{ij}.
\end{align}
If flow $f$ is LB-regulated, by Lemma~\ref{lemma:waiting_time} we have
\begin{align} \label{eq:CBFS_response_2a}
W(n,f,i,j,x) =T_{ij}^x + \frac{b_{ij}^{\mathrm{tot},x}-l_n}{R_{ij}^x},
\end{align} 

thus in this case
\begin{align} \label{eq:CBFS_response_2a}
D_n-A_n
&\leq T_{ij}^x + \frac{b_{ij}^{\mathrm{tot},x}}{R_{ij}^x}+l_n\left(\frac{1}{c_{ij}}-\frac{1}{R^x_{ij}}\right)+T^{\mathrm{var},\max}_{ij}
\\
&\leq T_{ij}^x + \frac{b_{ij}^{\mathrm{tot},x}}{R_{ij}^x}+M_f\left(\frac{1}{c_{ij}}-\frac{1}{R^x_{ij}}\right)+T^{\mathrm{var},\max}_{ij},
\end{align}
because $R^x_{ij}<c_{ij}$. This shows Eq. \eqref{eq:CBFS_response_time} in this case.

If flow $f$ is LRQ-regulated, by Lemma~\ref{lemma:waiting_time} we have
\begin{align} \label{eq:CBFS_response_2a}
W(n,f,i,j,x) =T_{ij}^x + \frac{b_{ij}^{\mathrm{tot},x}-L_f}{R_{ij}^x},
\end{align}

thus in this case
\begin{align} \label{eq:CBFS_response_2a}
D_n-A_n
&\leq T_{ij}^x + \frac{b_{ij}^{\mathrm{tot},x}-L_f}{R_{ij}^x}+\frac{l_n}{c_{ij}}+T^{\mathrm{var},\max}_{ij}
\\
&\leq T_{ij}^x + \frac{b_{ij}^{\mathrm{tot},x}-L_f}{R_{ij}^x}+\frac{L_f}{c_{ij}}+T^{\mathrm{var},\max}_{ij}.
\end{align}
This shows Eq. \eqref{eq:CBFS_response_time} in this case.

Now we show that the bound is tight via analyzing an example illustrated in Fig. \ref{fig:waiting_tightness}. For the ease of presentation, we drop the indices $i, j$ from the terms, $R_{ij}^x$, $T_{ij}^x$, $b_{ij}^{\text{tot},x}$, $L_{ij}^x$, $\bar{L}_{ij}^x$, $I_{ij}^x$, $S_{ij}^x$, $b_{ij}$, $r_{ij}$, $c_{ij}$.
We assume $L^E=\bar{L} =\bar{L}^A $.
At time $s_0$, a packet of BE with size $L^E$ arrives and starts being transmitted. At time $s_0^+$, a burst $b$ of CDT traffic arrives and then for time $t>s_0^+$, CDT traffic continues to arrive with rate $r$ up to the time $s_1$. The transmission of CDT traffic at time $s_0^+$ is blocked by the transmission of the BE packet, since it is assumed that lower priority queues are non-preemptive. At time $s=s_0+\frac{L^E}{c}$, CDT traffic has accumulated a backlog of $b+r\frac{L^E}{c}$ and starts its transmission.

The transmission time of this amount of CDT traffic is $\tau = \frac{b+r\frac{L^E}{c}}{c}$. During $\tau$, according to arrival curve of CDT, $r\tau$ bits are added to CDT queue. Similarly, during $r\tau$,
$r^2 \tau$ bits are added to the CDT queue, etc., and this continues until the time that the CDT queue empties, $s_1$. At $s_1$, class A can start transmitting its packets. The time interval between $s$ and $s_1$ can be calculated as follows:
\begin{align}
T_{s\rightarrow s_1} &= \tau + \frac{r\tau}{c} + \frac{r(\frac{r\tau}{c})}{c}+...=\tau(1+\frac{r}{c}+(\frac{r}{c})^2+...)\\
& = \frac{c}{c-r}\tau=\frac{c}{c-r} (\frac{b+r\frac{L^E}{c}}{c})=\frac{b+r\frac{L^E}{c}}{c-r}.
\end{align}
 At time $s^+$, there is also arrival of a burst of class $A$ traffic with $b^{\text{tot},A}=L^A+l_2$ and assuming that $l_2$ is a packet of flow $f$ that comes second in the burst.
 We are interested in the response time of the packet of this flow, $f$, in the CBFS of switch $i$. Finally, at time $s_3^-$ a packet of BE arrives and starts being transmitted, as class $A$ has still negative credit and the backlog of CDT is zero. Finally, at time $s_4^-$, CDT traffic arrives according to its arrival curve (i.e., $rT_{s_1\rightarrow s_4}$ if assuming that $rT_{s_1\rightarrow s_4}<b$, with $T_{s_1\rightarrow s_4}=\frac{1}{c}(L^A(\frac{I^A-S^A}{I^A})+\bar{L}^A)$ as shown in Fig. \ref{fig:waiting_tightness}) and starts being transmitted. Similarly with $T_{s\rightarrow s_1}$, we can compute that $T_{s_4\rightarrow s_5} = \frac{rT_{s_1\rightarrow s_4}}{c-r}$.

Due to the priority of CDT over class $A$ and of the latter over BE along with the traffic arrivals described above, the credit evolution of class $A$ is shown in Fig. \ref{fig:waiting_tightness}. The transmission of the packet $l_2$ of flow $f$ can only start at time $s_5$. As a result, assuming that $T^{\text{var,max}}_{ij}=0$, the CBFS response time of this packet based on this example is:
\begin{align}
S_{\mathrm{exp}}&=\frac{1}{c-r}\Bigl(b+\frac{r \bar{L}}{c} \Bigr)+T_{s_1\rightarrow s_4}+\frac{rT_{s_1\rightarrow s_4}}{c-r}+\frac{l_2}{c}\nonumber \\&=
\frac{1}{c-r}\Bigl(b+\frac{r \bar{L}}{c}+\bar{L}^A \Bigr)+
\frac{1}{c-r}\frac{L^A(I^A-S^A)}{I^A}+\frac{l_2}{c}.
\end{align}
The theoretical bound on the response time of this packet is,
\begin{align}
&S(f,i,j,A)=T^A+ \frac{b^{\mathrm{tot},A}- \psi^{f}}{R^A}+ \frac{\psi^{f}}{c}\nonumber \\&
=\frac{1}{c-r}\Bigl(b+\frac{r \bar{L}}{c}+\bar{L}^A \Bigr)+
\frac{1}{c-r}\frac{(I^A-S^A)}{I^A} (L^A+l_2- \psi^{f})\nonumber \\&+ \frac{\psi^{f}}{c}.
\end{align}

From our assumptions on this example, $\bar{L}^A=\bar{L}$. Since $S(f,i,j,A)$ is proven to be the bound, it holds that
\begin{align}
&S(f,i,j,A) \geq S_{\mathrm{exp}},
\end{align}and for tightness it should hold
\begin{align}
S(f,i,j,A) =S_{\mathrm{exp}}.
\end{align}
Therefore, for tightness, we require that $l_2=\psi^{f}$. 
If the flow is LB-regulated, the bound is tight if $l_2=\psi^{f}=M_{f}$, and if it is LRQ-regulated, the bound is tight with $l_2=\psi^{f}=L_{f}$. In the case of LRQ, $l_2=L_f$ implies $b^{\text{tot},A}=L^A + L_{f}$; i.e. for LRQ there should be at least two flows in queue of class $x$ so that the bound for the CBFS waiting time is achieved (since in case of LRQ flow $f$ has an arrival curve at the CBFS with burstiness $L_f$).
\end{proof}


\begin{lemma}\label{lemma:waiting_time}
Consider the $n^{th}$ packet in the CBFS queue of switch $i$, assuming it belongs to flow $f$ of class $x\in \{A,B\}$, going from switch $i$ to $j$. An upper bound on the waiting time of this packet of flow $f$ in the CBFS of switch $i$ (following the interleaved regulator of $i$) is:
\begin{align}\label{eq:waiting_time}
W(n,f,i,j,x) = T_{ij}^x + \frac{b_{ij}^{\mathrm{tot},x}-\psi^f_n}{R_{ij}^x},
\end{align} \noindent where the parameter $\psi^f_n$ depends on the $n^{th}$ packet that belongs to flow $f$ and the type of regulator, namely, for LRQ: $\psi^f_n = L_f$ and for LB: $\psi^f_n = l_n$, where $l_{n}$ stands for the packet size of the $n^{th}$ packet.
\end{lemma}

\begin{proof}
We do the proof for a flow of class $x$ going from switch $i$ to $j$. For the ease of presentation, we drop the indices $i, j, x$ from the terms, $R_{ij}^x$, $T_{ij}^x$, $b_{ij}^{\text{tot},x}$. Assume that, $l_{k}$ stands for the packet size of the $k^{th}$ packet in the AVB FIFO queue at the CBFS.


 From Lemma \ref{lemma:min_service_curve}, the following equation is true for a service curve $\beta(t)$ and $m\leq n$,
 \begin{align}\label{eq:scheduler_del_4}
 \sum_{k=m}^{n-1} l_k \geq \beta(Q_n-A_m).
 \end{align}
Let $\beta^{\uparrow}(t)$ be the upper pseudo-inverse of a service curve \cite{liebeherr_duality_2017} $\beta(t)$, defined as $\beta^{\uparrow}(t) =\sup \{s\geq 0 | \beta(s)\leq t\}=inf \{s\geq 0 | \beta(s)>t\}$. Then, by taking upper pseudo-inverse in both sides of Eq. \eqref{eq:scheduler_del_4},
 \begin{align}\label{eq:scheduler_del_5}
 Q_n-A_m \leq \beta^{\uparrow}\Big(\sum_{k=m}^{n-1} l_k\Big).
 \end{align}
 Therefore, $Q_n$ satisfies,
 \begin{align}\label{eq:scheduler_del_6}
 Q_n \leq \max_{m\leq n} \Bigg\{A_m+\beta^{\uparrow}\Big(\sum_{k=m}^{n-1} l_k\Big)\Bigg\}.
 \end{align} \noindent Let $F(k)$ be the flow id of $k^{th}$ packet and $f$ be the flow id of the head-of-line packet that is currently being served in the CBFS, say packet $n$. Then, we can express the sum of packet lengths as,
 \begin{align}\label{eq:scheduler_del_7}
 \sum_{k=m}^{n-1} l_k = \sum_{f'\neq f} \Bigg[\sum_{k=m}^{n-1} 1_{\{F(k)=f'\}}l_k\Bigg] + \sum_{k=m}^{n-1} 1_{\{F(k)=f\}}l_k.
 \end{align}



 As the input flows are regulated, for flows $f' \neq f$ there is an arrival curve $\alpha_{f'}$ that satisfies,

 \begin{align}\label{eq:scheduler_del_8}
 \sum_{k=m}^{n-1}1_{\{F(k)=f'\}} l_k \leq \alpha_{f'}(A_n-A_m).
 \end{align}
 \noindent Also, for flow $f$ that is being served now, we show that
 \begin{align}\label{eq:scheduler_del_9}
 \sum_{k=m}^{n-1}1_{\{F(k)=f\}} l_k \leq \alpha_{f}(A_n-A_m)-\psi^{f}_n,
 \end{align}
 \noindent where $\psi^{f}_n$ is the parameter defined in the statement of the lemma. For the considered regulators, LB and LRQ, $\alpha_{f}, \forall f$ are of the form $r_f t +b_f$ where for the LRQ $b_f=L_f$. $\psi^{f}_n$ is computed as follows. For the LB, $\sum_{k=m}^{n-1}1_{\{F(k)=f\}} l_k + l_n \leq \alpha_{f}(A_n-A_m)$, thus $\psi^{f}_n=l_n$. For the LRQ, based on its definition,
 \begin{align}\label{eq:scheduler_del_9b}
A_{k+1}-A_{k}\geq \frac{l_{k}}{r_{F(k)}}.
 \end{align} Therefore,
 \begin{align}\label{eq:scheduler_del_10b}
\sum_{k=m}^{n-1}1_{\{F(k)=f\}}  \frac{l_{k}}{r_{f}} \leq A_n-A_m,
 \end{align} or
 \begin{align}\label{eq:scheduler_del_9c}
\sum_{k=m}^{n-1}1_{\{F(k)=f\}}  l_{k} \leq r_f (A_n-A_m).
 \end{align} Since for LRQ, $\alpha_{f}(t)=r_f t +L_f$, we obtain $\psi^{f}_n=L_f$.

 By utilizing Eqs. (\ref{eq:scheduler_del_7}), (\ref{eq:scheduler_del_8}) and (\ref{eq:scheduler_del_9}) in Eq. (\ref{eq:scheduler_del_6}), we obtain,
 \begin{align}\label{eq:scheduler_del_10}
 Q_n \leq \max_{m\leq n} \Bigg\{A_m+\beta^{\uparrow}\Big(\sum_{f'} \alpha_{f'}(A_n-A_m) -\psi^{f}_n\Big)\Bigg\}.
 \end{align}
 By setting $A_n-A_m=t \geq 0$, we obtain,
 \begin{align}
 Q_n - A_n &\leq \sup_{t \geq 0} \Bigg\{-t+\beta^{\uparrow}\Big(\sum_{f'} \alpha_{f'}(t) -\psi^{f}_n\Big)\Bigg\} \nonumber \\
&= \sup_{t \geq 0} \Bigg\{-t+\frac{\sum_{f'} \alpha_{f'}(t) -\psi^{f}_n}{R}+T\Bigg\}
 \\
 &=\sup_{t \geq 0} \Bigg\{-t+\frac{\sum_{f'} \alpha_{f'}(t)}{R}\Bigg\}-\frac{\psi^{f}_n}{R}+T
 \\
 &\leq \frac{\sum_{f'} \alpha_{f'}(0)}{R}-\frac{\psi^{f}_n}{R}+T
 \\\label{eq:scheduler_del_11}
 &=\frac{b^{\mathrm{tot}}}{R}-\frac{\psi^{f}_n}{R}+T.
 \end{align}

\begin{figure}
	\centering
	\includegraphics[width=1 \linewidth]{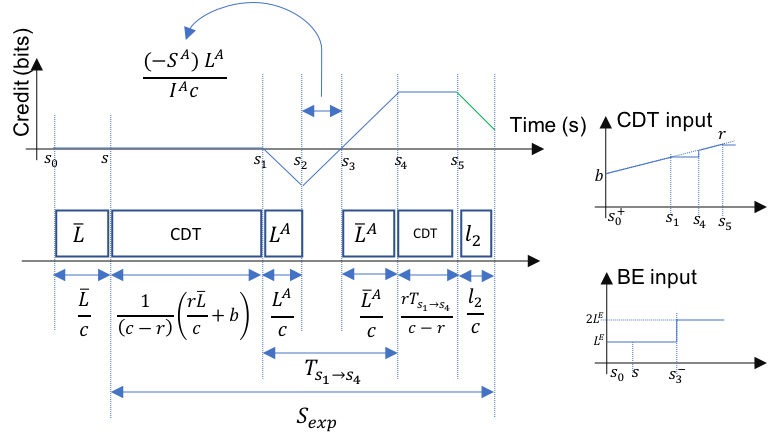}
	\caption{Scenario to have tight bound on response time of CBFS for class A}
	\label{fig:waiting_tightness}
\end{figure}

The right part of Eq. \eqref{eq:scheduler_del_11} gives $W(n,f,i,j,x)$ in the statement of the lemma.

\end{proof}


\begin{lemma}\label{lemma:min_service_curve}
If $\beta$ is a service curve, then for each packet $n$, there exists an $m \leq n $ such that, $\beta(Q_n - A_m) \leq \sum_{k = m}^{n-1} l_{k}$, where $l_k$ is the size of the $k^{th}$ packet.
\end{lemma}
\begin{proof}
From the definition of a service curve \cite{le_boudec_network_2001}, we have\\
\begin{equation}
\forall~t,~\exists~s: O(t)~\geq~N(s)~+~\beta(t-s), \label{eq:min_service_def}
\end{equation}
\noindent where $O(t)$ is the number of bits that have been served until time $t^-$ and $N(t)$ is the number of bits that have arrived until time $t$ and $\beta(\cdot) \geq 0$.
\noindent Let $m$ such that, $A_{m-1} < s \leq A_{m}$, then,
\begin{equation}
N(s) = l_1+ \hdots + l_{m-1}. \label{eq:beta_m}
\end{equation}
\noindent Then, for some $\delta \geq 0$, we have
\begin{align}
s &= A_m - \delta\label{eq:def_delta}\\
\implies \beta(t - s) &= \beta(t - A_m + \delta) \geq \beta(t - A_m).\label{eq:beta1}
\end{align}
By replacing Eq. \eqref{eq:def_delta} and \eqref{eq:beta1} in \eqref{eq:min_service_def}, we obtain:
\begin{equation}
O(t)~\geq~N(A_m - \delta)~+~\beta(t-A_m).\label{eq:after_putting_m}
\end{equation}
By setting $t = Q_n^+$, we obtain $t$ as the time instant we begin to serve packet $n$. This means, that all packets before $n$ have already been served. Thus,
\begin{align}
O(t)~&=~l_1 +\hdots+l_{n-1}.\label{eq:output_t}
\end{align}
\noindent By replacing Eqs. \eqref{eq:beta_m}, \eqref{eq:output_t} in \eqref{eq:after_putting_m}, we obtain:
\begin{align}
l_1 +\hdots+l_{n-1}~&\geq~l_1 + \hdots + l_{m-1} + \beta(Q_n - A_m)\nonumber\\
\implies \beta(Q_n - A_m)~&\leq l_m + \hdots + l_{n-1}.
\end{align}
If $m > n$, we obtain $\beta(Q_n - A_m) < 0$, which is a contradiction.
\end{proof} 

\section{Proof of Theorem \ref{thm:response_time}}\label{app:3}
\begin{proof}
$H(f,i,j,k,x)$ is an upper bound on $E_n - D'_n$ for all packets $n$ of flow $f$. Now
 
	\begin{align}
	E_n - D'_n = (E_n-A_n) -(D'_n-A_n),
	\end{align}
and 
\begin{align}
	E_n - A_n \leq C(i,j,k,x).
	\end{align}
Therefore, for a packet $n$ of flow $f$:
\begin{align}
	E_n - D'_n \leq C(i,j,k,x) - \inf_{n\in N_f}(D'_n-A_n),
	\end{align}
now
\begin{align}
	D'_n -A_n \geq \frac{M_f}{c_{ij}} + T_{ij}^{\mathrm{var,min}}+T_{ij}^{\text{proc, min}}.
	\end{align}
 
%
%
\end{proof}

\end{document}